\documentclass[10pt,aps,pra,twocolumn,floatfix,nofootinbib,superscriptaddress,longbibliography]{revtex4-2}
\usepackage{placeins}

\usepackage{bm,graphicx,mathrsfs,amsmath,amssymb,mathtools,makecell,bbm, amsthm,dsfont,color,nicefrac,framed,enumitem,tikz,physics,wrapfig,amsfonts,tcolorbox,times,txfonts,lipsum,nicematrix}
\usepackage[T1]{fontenc}
\usepackage[colorlinks=true,linkcolor=refcolor,citecolor=refcolor]{hyperref}

\hypersetup{urlcolor=refcolor}
\definecolor{equationcolor}{RGB}{108,153,224}
\definecolor{refcolor}{RGB}{214,86,86}
\definecolor{changescolor}{RGB}{188, 104, 104}

\makeatletter
\def\blfootnote{\gdef\@thefnmark{}\@footnotetext}
\makeatother

\renewcommand{\v}[1]{\ensuremath{\boldsymbol #1}}
\newcommand{\ms}[1]{\textsf{#1}}

\newcommand{\iden}{\mathbbm{1}}

\newcommand{\E}[1]{\mathcal{E}}

\def\E{ {\cal E} }

\newenvironment{sproof}{%
  \proof}{\endproof}

\newtheorem{thm}{Theorem}
\newtheorem*{res*}{Result}

\newtheorem{lem}[thm]{Lemma}

\newtheorem{cor}[thm]{Corollary}

\begin{document}

\title{Robust Bell Nonlocality from Gottesman–Kitaev–Preskill States}

\author{Xiaotian Yang}
	\affiliation{Center for Macroscopic Quantum States bigQ, Department of Physics,
Technical University of Denmark, Fysikvej 307, 2800 Kgs. Lyngby, Denmark}
\author{Santiago Zamora}
\affiliation{Physics Department, Federal University of Rio Grande do Norte, Natal, 59072-970, Rio Grande do Norte, Brazil}
\affiliation{International Institute of Physics, Federal University of Rio Grande do Norte, 59078-970, Natal, RN, Brazil}
\author{Rafael Chaves}
    \affiliation{International Institute of Physics, Federal University of Rio Grande do Norte, 59078-970, Natal, RN, Brazil}
    \affiliation{School of Science and Technology, Federal University of Rio Grande do Norte, Natal, Brazil}
\author{Ulrik L. Andersen}
	\affiliation{Center for Macroscopic Quantum States bigQ, Department of Physics,
Technical University of Denmark, Fysikvej 307, 2800 Kgs. Lyngby, Denmark}
\author{Jonatan Bohr Brask}
    \affiliation{Center for Macroscopic Quantum States bigQ, Department of Physics,
Technical University of Denmark, Fysikvej 307, 2800 Kgs. Lyngby, Denmark}
\author{A. de Oliveira Junior}
\email{alexssandredeoliveira@gmail.com}
	\affiliation{Center for Macroscopic Quantum States bigQ, Department of Physics,
Technical University of Denmark, Fysikvej 307, 2800 Kgs. Lyngby, Denmark}
\date{\today}

\begin{abstract}
    Bell tests based on homodyne detection are strongly constrained in continuous-variable systems. Can Gottesman–Kitaev–Preskill (GKP) encoding turn homodyne detection into a practical tool for revealing Bell nonlocality? We consider a physically motivated model in which each party performs homodyne detection and digitizes the continuous outcome via a fixed periodic binning, corresponding to logical Pauli measurements. Within this framework, we derive a bipartite no-go: CHSH cannot be violated for Bell-pair states. Moving beyond two parties, we show that finitely squeezed GKP-encoded GHZ and W states nevertheless exhibit strong multipartite nonlocality, violating multipartite Bell inequalities with homodyne-only readout. We quantify the required squeezing thresholds and robustness to loss, providing a route toward homodyne-based Bell tests in continuous-variable systems.
\end{abstract}

\maketitle

\emph{Introduction--} The violation of Bell inequalities decrees the death knell for local realism~\cite{Einstein1935,Bell1964,Bell1966OnTP,Brunner2014}, while simultaneously providing an operational witness for certifying non-classicality in quantum technologies~\cite{Ekert1991,Valerio2007,Pironio2010,Vazirani2014}. Making this witness experimentally meaningful, however, depends on the measurement scheme: if many events are lost due to detector inefficiencies, local-realistic models can reproduce the observed statistics, reopening the detection loophole~\cite{Larsson2014}. This raises the question of which measurements can reveal Bell nonlocality under realistic constraints.

Homodyne detection is a natural candidate in this respect, as it can reach very high efficiencies and is therefore attractive for achieving loophole-free Bell tests~\cite{Patron2004,Eberhard93}. However, it also comes with a fundamental limitation because Bell nonlocality from Gaussian measurements on Gaussian states is impossible~\cite{Bell2004,JOHANSEN1997173}. Hence, non-Gaussian states must be used. While photon-subtracted states or optical cat states can circumvent this restriction~\cite{Wenger2003,GarciaPatron2004,GarciaPatron2005}, their non-Gaussian features are fragile under loss, limiting their viability in realistic implementations. Such tests also require discretizing the continuous measurement outcomes into a finite set of results~\cite{GarciaPatron2004,GarciaPatron2005}. This step is unavoidable when testing standard Bell inequalities, but it has nontrivial consequences: in many proposals the binning is optimized for a given state or measurement setting, which can significantly affect both the size and the robustness of the observed violation~\cite{Wenger2003,Acin2009,Oudot2024,Carlos2026}. A fundamentally different and more robust route to address these issues is provided by the Gottesman--Kitaev--Preskill (GKP) encoding~\cite{GKP2001,Glancy2006,Grimsmo2021,Tzitrin2020}.

\begin{figure}[t]
    \centering
    \includegraphics{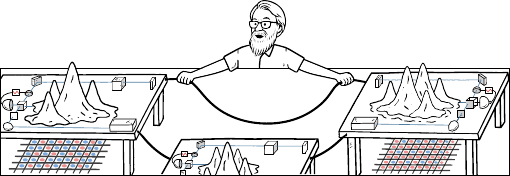}
    \caption{\emph{Bell nonlocality with Gottesman–Kitaev–Preskill States}. Nonlocal correlations can be pulled out of continuous-variable systems using GKP states and only homodyne detection.}
        \label{F-Bells-violation}
\end{figure}

Originally conceived for fault-tolerant quantum computing, GKP states encode qubits in one bosonic mode using a grid structure in phase space. This geometry leverages the oscillator’s infinite-dimensional Hilbert space to provide an inherent redundancy that would otherwise require hundreds of discrete physical qubits~\cite{Noh2020}. Crucially, the GKP framework allows for universal operations, error correction, and measurements using only Gaussian interactions~\cite{Menicucci2014}. Driven by the need for scalable quantum architectures, GKP states have been realised across diverse platforms, including trapped ions~\cite{Flhmann2019,Hastrup2021}, microwave cavities~\cite{CampagneIbarcq2020}, and integrated photonic sources~\cite{Hussler2025,Larsen2025Nature}.

\begin{figure*}
    \centering
    \includegraphics{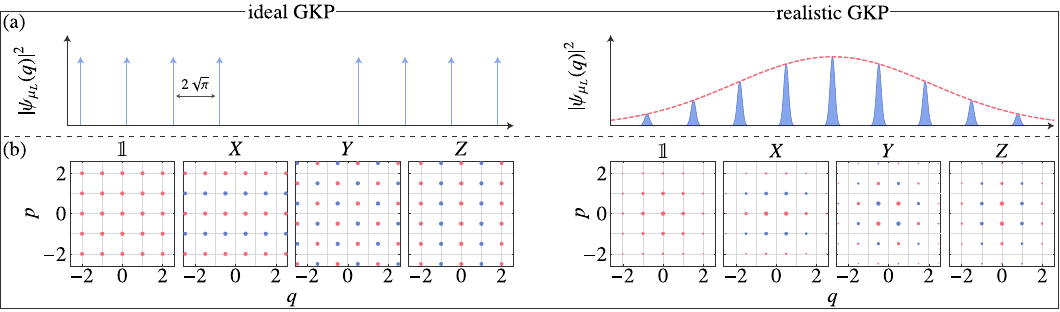}
    \caption{\emph{Ideal vs finite-energy GKP structure in phase space}. Comparison between ideal and finite-energy (``realistic'') square GKP encodings in phase space and in the position quadrature. (a) sketch of the position-space marginal distribution showing that the ideal GKP code words form an infinitely sharp Dirac-comb lattice, while finite-energy GKP states replace each spike by a narrow Gaussian peak and suppress distant peaks under a broad Gaussian envelope (dashed curve), giving a normalizable state with finite mean energy. (b) Phase-space dot/contour plots of the Wigner functions associated with the logical GKP Pauli operators $\iden, X, Y$ and $Z$. Red (blue) marks positive (negative) regions, with the dot size indicating the relative magnitude of each lattice peak.}
    \label{F-wigner-function-GKP-sketch}
\end{figure*}
Can the GKP encoding unlock robust Bell violations with homodyning? A first step toward bridging non-locality and the GKP code was taken via a device-independent quantum cryptography protocol~\cite{Marshall2014}. Alice generates an entangled GKP state and shares one half with Bob to establish a secret key based on a CHSH violation. The drawback of this approach is that the violation requires a non-Gaussian T-gate, whose experimental implementation is challenging. This leads us to ask: given a GKP encoding, can we violate a Bell inequality using only Gaussian operations? Realistic GKP states are approximate, with finite squeezing and sensitivity to loss that can blur their grid structure. Whether this residual, noise-degraded non-Gaussianity is sufficient to overcome the homodyne no-go and yield an observable Bell violation has been unclear. In this work, we show that GKP states enable sizeable Bell violations using only homodyne detection, even in the presence of finite squeezing and noise~(see Fig.~\ref{F-Bells-violation} for a summary).

\emph{Big picture--} We characterise Bell nonlocality under a strict measurement constraint: each party performs homodyne detection of a quadrature, and the continuous outcome is converted into a binary result by a fixed binning rule. While the shared state may be non‑Gaussian and highly entangled, this restriction yields a fundamental bipartite obstruction: Pauli measurements\footnote{Here ``Pauli measurements'' means measurements of the three Pauli observables $X, Y, Z$ (not general qubit observables).} on Bell states generate only local correlations in the CHSH scenario (see Appendix~\ref{App:nogo-chsh-pauli}). This motivates our focus on multipartite Bell inequalities, for which Pauli measurements already lead to strong (and often optimal) violations. In quantum theory, Bell violations arise from entanglement probed by incompatible local measurements. We now introduce the two main ingredients of our continuous-variable implementation: the shared entangled resource and the local measurement settings. The resource is provided by embedding a logical qubit into a single bosonic mode using GKP grid states, while the measurement choices are implemented using homodyne detection followed by a fixed binary post-processing (binning) of the continuous outcome.

\emph{Gottesman--Kitaev--Preskill as a resource--} We use the square-lattice Gottesman--Kitaev--Preskill (GKP) encoding, which embeds a logical qubit into a single bosonic mode~\cite{GKP2001,Grimsmo2021}. In the $\hat q$-eigenbasis, the ideal (non-normalisable) logical codewords are Dirac-comb superpositions,
\begin{align}\label{Eq:GKP-basis}
\ket{0_{L}} &\propto \sum_{s\in\mathbb{Z}} \ket{2s\sqrt{\pi}}_{q}, \\
\ket{1_{L}} &\propto \sum_{s\in\mathbb{Z}} \ket{(2s+1)\sqrt{\pi}}_{q},
\end{align}
where $\ket{\xi}_q$ denotes the (generalised) eigenket of $\hat q$ with eigenvalue $\xi$. These ideal states are unphysical,
as they require infinite energy.

To model realistic GKP states, we adopt the finite-energy approximation generated by the energy filter $e^{-\epsilon\hat n}$~\cite{Matsuura2020,Girvin2020,Hastrup2022,Hastrup2023}, where $\hat n=\hat a^\dagger\hat a=\tfrac12(\hat q^2+\hat p^2-1)$. This regularisation replaces the Dirac comb by a lattice of narrow Gaussian peaks under a broader Gaussian envelope (see Fig.~\hyperref[F-wigner-function-GKP-sketch]{\ref{F-wigner-function-GKP-sketch}a} for an illustration and Appendix~\ref{App-finite-energy-GKP-state} for details). For $\epsilon>0$ the resulting codewords are normalisable and denoted by $\ket{\mu_L^\epsilon}$ with $\mu\in\{0,1\}$. The variance of the individual peaks (in the quadrature marginals) is $\sigma^2=\tfrac12\tanh\epsilon$. We quantify the grid quality by the effective noise/squeezing parameter (in dB) $r_{\mathrm{dB}} := -10\log_{10}\bigl(\tanh\epsilon\bigr)$, so that larger $r_{\mathrm{dB}}$ corresponds to sharper peaks. The same regularisation also contracts the lattice by a factor $\sech\epsilon$ and introduces an envelope variance that ensures finite mean energy.

An arbitrary logical $N$-qubit state $\rho_L$ is encoded into $N$ bosonic modes via the map $V_\epsilon^{\otimes N}$, where \mbox{$V_\epsilon := \ket{0_L^\epsilon}\!\Bigl\langle0\Bigl| + \ket{1_L^\epsilon}\!\Bigl\langle1\Bigl|$}. Since the finite-energy codewords are not exactly orthogonal, $V_\epsilon$ is generally not an isometry. Consequently the map $\rho_L\mapsto V_\epsilon^{\otimes N}\rho_LV_\epsilon^{\dagger\otimes N}$ is not trace-preserving. To write the encoded state explicitly, we define the (unnormalized) encoded Pauli operators $\sigma_k^\epsilon := V_\epsilon \sigma_k V_\epsilon^\dagger$ with $k=0,1,2,3$, where $\sigma_0=\iden$ and $\sigma_{1,2,3} = (X,Y,Z)$ are the usual Pauli matrices. The resulting physical $N$-mode state is 
\begin{equation}\label{Eq:N-multipartite-state}
    \rho_L^\epsilon =\frac{1}{\mathcal N}\sum_{\v k\in\{0,1,2,3\}^N} c_{\v k}\; \sigma_{k_1}^\epsilon\otimes\cdots\otimes\sigma_{k_N}^\epsilon,
\end{equation}
where $c_{\v k}:=\tr\left[\rho_L\left(\sigma_{k_1}\otimes\cdots\otimes\sigma_{k_N}\right)\right]$ and the normalization $\mathcal N$ ensures $\tr(\rho_L^\epsilon)=1$.

In the energy-filter model, the operators $\sigma_k^\epsilon$ inherit the ``grid'' structure of GKP states: in phase space they can be written as a sum of Gaussian peaks centred on a (contracted) square lattice, each peak weighted by a broad Gaussian envelope (see Fig.~\hyperref[F-wigner-function-GKP-sketch]{\ref{F-wigner-function-GKP-sketch}b} and Appendix~\ref{App-finite-energy-GKP-state}). Different $k$ correspond to different parity sublattices of this grid and to an alternating $\pm$ interference pattern between peaks. To encode this compactly, we label lattice sites by $\v m=(m_1,m_2)\in\mathbb Z^2$, define the envelope weight $c_{\v m}^\epsilon=\exp\left[-\frac{\pi}{4}\tanh(\epsilon)\,(m_1^2+m_2^2)\right]$, and introduce the lattice cosets
\begin{align}\label{Eq:Mk-sets}
\mathcal M_k = \Bigl\{\v m \equiv \qty(k,\tfrac{k(k+1)}{2})\pmod 2\Bigr\}
\end{align}
together with sign patterns $s_k(\v m)\in\{\pm1\}$ on $\mathcal M_k$. One convenient compact form is $s_k(\v m)=(-1)^{\frac12\left[\delta_{k,1}m_2+\delta_{k,3}m_1+\delta_{k,2}(m_1+m_2)\right]}$, where $\delta_{k,\ell}$ denotes the Kronecker delta.

This encoding is naturally compatible with homodyne-based readout. Homodyne detection measures the rotated quadrature $\hat q_\theta\!=\!\hat q\cos\theta+\hat p\sin\theta$ and returns a real outcome $q_\theta\in\mathbb{R}$. A measurement setting $x$ fixes an angle $\theta_x$, and the corresponding two-outcome POVM elements are \mbox{$M_o^{(x)} := \int_{R_o} \mathrm{d}q_{\theta_x} \ketbra{q_{\theta_x}}$} for $o=0,1$, where the binning regions defined as the periodic unions of half-open intervals
\begin{equation}\label{Eq:binning}
    R_o := \bigcup_{s\in\mathbb Z}\left[(2s+o)\sqrt{\pi}-\frac{\sqrt{\pi}}{2}, (2s+o)\sqrt{\pi}+\frac{\sqrt{\pi}}{2}\right).
\end{equation}
We assign the bit outcome $o$ according to whether $q_{\theta_x}\in R_o$. In the ideal GKP limit, the binning sets $R_o$ are aligned with the code lattice in the measured quadrature: measuring $\hat q$ (i.e., $\theta=0$) separates the two position-comb cosets and thus realises an effective logical $Z$, while measuring $\hat p$ (i.e., $\theta=\tfrac\pi2$) gives the complementary effective logical $X$ in the conjugate quadrature. A diagonal homodyne angle $\theta=\tfrac\pi4$ probes the lattice along the $q=p$ direction and yields the corresponding effective logical $Y$. Accordingly, we take $Z$ to correspond to $\theta_Z=0$, $X$ to $\theta_X=\tfrac{\pi}{2}$, and $Y$ to $\theta_Y=\tfrac{\pi}{4}$.

\emph{Multipartite Bell-test \& GKP encoding--} We consider an $N$-party Bell experiment with $N\ge 3$, where each party holds a single bosonic mode encoding one logical qubit. In each run, party $j$ receives an input $x_j\in\{0,1\}$ selecting one of two homodyne angles $\theta_{x_j}$ and measures the corresponding quadrature $\hat q_{\theta_{x_j}}$. The continuous outcome is then discretized into a bit $o_j\in\{0,1\}$ using the fixed periodic binning rule defined in Eq.~\eqref{Eq:binning}. The experiment is thus described by the conditional probabilities $p(\v o|\v x)$. A behaviour $p(\v o|\v x)$ is Bell-local if it admits a local-hidden-variable decomposition \mbox{$p(\v o|\v x)=\int \mathrm{d}\lambda\, p(\lambda)\prod_{j=1}^N p_j(o_j\mid x_j,\lambda)$}, and any violation of a Bell inequality derived from this assumption certifies multipartite Bell nonlocality.

Our central technical ingredient is an explicit expression for the full $N$-party behaviour generated by any finite-energy GKP-encoded state under this measurement scheme. Let $\rho_L^\epsilon$ be the encoded state in Eq.~\eqref{Eq:N-multipartite-state} and define the product POVM $M_{\v o}^{(\v x)}:=\bigotimes_{j=1}^N M_{o_j}^{(x_j)}$, where $M_{o}^{(x)}$ is the single-mode binned-homodyne element for angle $\theta_x$. Using the Pauli-string expansion of $\rho_L^\epsilon$, the probability factorizes as
\begin{equation}\label{Eq:factorisation}
p(\v o|\v x)=\tr\left[\rho_L^\epsilon\, M_{\v o}^{(\v x)}\right] =\frac{1}{\mathcal N}\sum_{\v k\in\{0,1,2,3\}^N} c_{\v k}\prod_{j=1}^N t_{k_j,o_j}^{(x_j)},
\end{equation}
where $c_{\v k}$ are the logical Pauli coefficients from Eq.~\eqref{Eq:N-multipartite-state} and $t_{k,o}^{(x)}:=\tr\left(M_o^{(x)}\sigma_k^\epsilon\right)$ is a single-mode overlap. This factorisation reduces the $N$-mode behaviour to single-mode overlaps. Their explicit form follows from the Gaussian-mixture structure of the GKP Pauli operators' Wigner functions: 

\begin{res*}[Single-mode overlap]\label{res:single-overlap}
For the finite-energy GKP Pauli operators $\sigma_k^\epsilon$ and the two-outcome binned-homodyne POVM, the single-mode overlap $t_{k,o}^{(x)}$ admit the absolutely convergent error-function series
\begin{equation}\label{Eq:main-overlap}
t_{k,o}^{(x)}=(-1)^{\delta_{k,2}} \sum_{\v m\in\mathcal M_k} c_{\v m}^\epsilon\, s_k(\v m) \mathcal B_o\bigl(\mu_{\theta_x,\v m}\bigr),
\end{equation}
where $\mu_{\theta_x,\v m}=\sech(\epsilon) \frac{\sqrt{\pi}}{2}\bigl(m_1\cos\theta_x+m_2\sin\theta_x\bigr)$ is the quadrature projection of the lattice point $\v m=(m_1,m_2)\in\mathbb Z^2$, and the binning function $\mathcal B_o(\mu)$ is
\begin{equation}\label{eq:binning-function}
\mathcal B_o(\mu) =\frac12 \sum_{\ell\in\mathbb Z}\left[ \erf\bigl(a_{\ell,o}^+(\mu)\bigr)-\erf\bigl(a_{\ell,o}^-(\mu)\bigr) \right],
\end{equation}
with $a_{\ell,o}^\pm(\mu):=[\bigl(2\ell+o\pm\tfrac12\bigr)\sqrt{\pi}-\mu]/(\sqrt{2}\sigma)$ and $\sigma$ being the peak-width parameter of the finite-energy model.
\end{res*}
\begin{sproof}
    Since $M_o^{(x)}=\int_{R_o}\mathrm{d}q_{\theta_x}\ketbra{q_{\theta_x}}$, the overlap $t_{k,o}^{(x)}$ is the integral of the rotated quadrature marginal $\langle q_{\theta_x}|\sigma_k^\epsilon|q_{\theta_x}\rangle$ over $R_o$. Using the rotated Wigner-marginal identity, this marginal is a line integral of $W_{\sigma_k^\epsilon}$ along the conjugate quadrature $p_{\theta_x}$. Substituting the Gaussian-mixture form of $W_{\sigma_k^\epsilon}$ reduces the $p_{\theta_x}$ integration to a one-dimensional Gaussian for each lattice site, centred at $\mu_{\theta_x,\v m}$; integrating over the periodic bin then yields the error-function series $\mathcal B_o(\mu_{\theta_x,\v m})$. The lattice sum converges absolutely due to the Gaussian envelope $c_{\v m}^\epsilon$ (see Appendix~\ref{App-finite-energy-GKP-state}).
\end{sproof}

The above result reduces the full $N$-mode probability distribution to efficiently computable single-mode lattice sums, enabling numerics for multipartite inequalities under loss and finite squeezing.

\emph{Examples: GHZ \& W non-locality--} We now use our technical result to compute the binned-homodyne behaviour $p(\v o|\v x)$ and thus evaluate multipartite Bell witnesses tailored to two distinct entanglement structures. As logical resources we consider the $N$-qubit GHZ and $W$ states,
\begin{align}
\ket{\mathrm{GHZ}_N}_L &:= \frac{1}{\sqrt{2}}\bigl(\ket{0}_L^{\otimes N}+\ket{1}_L^{\otimes N}\bigr), \nonumber\\
\ket{W_N}_L &:= \frac{1}{\sqrt{N}}\sum_{j=1}^N \ket{\v e_j}_L ,
\end{align}
where $\ket{0}_L,\ket{1}_L$ are the logical GKP basis states and $\ket{\v e_j}_L$ denotes the computational-basis string with a single excitation at site $j$. Throughout, the physical (finite-squeezing) states used in the Bell test are the encoded versions $\rho_L^\epsilon$ obtained from these logical states via Eq.~\eqref{Eq:N-multipartite-state}. 
\begin{figure*}
    \centering
    \includegraphics{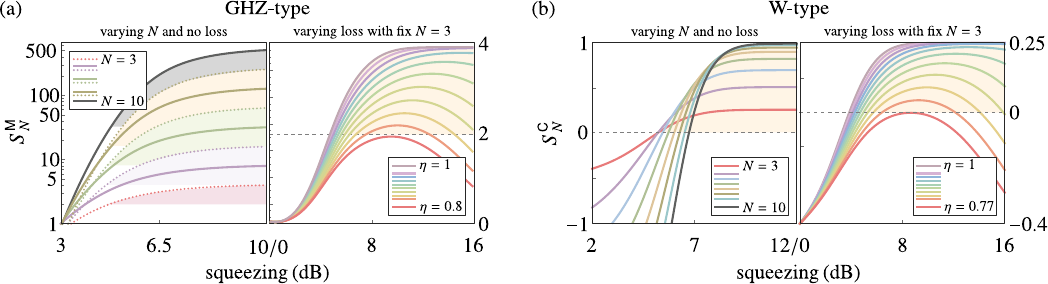}
    \caption{\emph{Multipartite Bell nonlocality with finitely squeezed GKP encoding under homodyne measurement.} (a) MABK Bell parameter $S_N^{\ms M}$ for encoded GHZ states. Left: lossless case for different numbers of parties $N$. Shaded regions indicate parameter regimes where the LHV bound $S_N^{\ms M}\le 2^{\lfloor N/2\rfloor}$ is violated. Odd $N$ are shown with dotted lines and even $N$ with solid lines, highlighting that some consecutive system sizes share the same classical bound. Right: $S_3^{\ms M}$ versus squeezing for different loss transmissivities $\eta$; the dashed horizontal line marks the LHV bound. (b) Cabello functional $S_N^{\ms C}$ for encoded $W$ states. Left: lossless case for varying $N$. Right: $N=3$ with varying $\eta$; the dashed horizontal line marks the local bound $S_N^{\ms C}\le 0$.}
    \label{F-MABK-W}
\end{figure*}

To witness the full-correlation nonlocality characteristic of GHZ states, we use the Mermin--Ardehali--Belinskii--Klyshko (MABK) family~\cite{Mermin1990,Ardehali1992,Belinskii1993}, which is optimally violated by GHZ states under Pauli measurements. Mapping bit outcomes to $\pm1$ variables $a_j:=(-1)^{o_j}$, we form the full correlators $E(\v x)=\langle\prod_{j=1}^N a_j\rangle$, computed from $p(\v o|\v x)$. The corresponding Bell parameter $S_N^{\ms M}$ is defined recursively in Appendix~\ref{App:bell-inequalities} and satisfies the local bound $S_N^{\ms M}\le 2^{\lfloor N/2\rfloor}$. In our implementation, we choose two homodyne angles corresponding to the effective logical settings $(Y,X)$.

To probe nonlocality associated with $W$-type entanglement, we use a Cabello-type probability inequality~\cite{Cabello2002,Heaney2011,chaves2011feasibility}, which is expressed directly in terms of selected joint probabilities. Restricting to two effective logical settings $(Z,X)$, we define the functional $S_N^{\mathcal{C}}$ using probabilities from only three classes of setting strings: \emph{(i)} all parties measure $Z$, \emph{(ii)} all parties measure $X$, and \emph{(iii)} for each unordered pair $(i,j)$, parties $i$ and $j$ measure $X$ while all others measure $Z$ (a total of $2+\binom{N}{2}$ configurations). The explicit expression is given in Appendix~\ref{App:bell-inequalities}. Any LHV model satisfies $S_N^{\ms C}\le 0$.
\begin{figure*}
    \centering
    \includegraphics{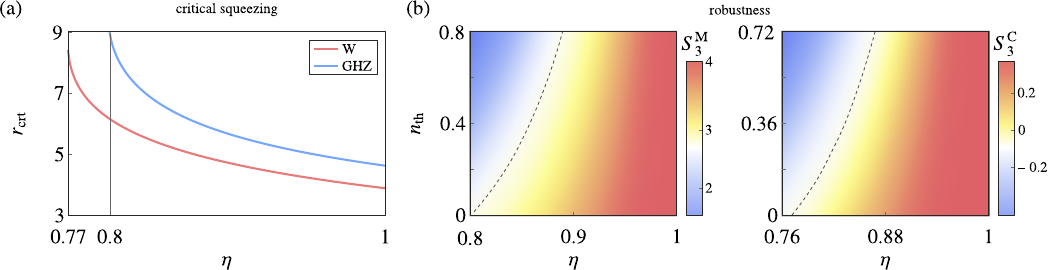}
    \caption{\emph{Critical squeezing and robustness to excess noise.} (a) Critical squeezing $r_{\mathrm{crit}}$ (dB) for observing a violation at $N=3$ versus loss transmissivity $\eta$, comparing GHZ/MABK (blue) and W/Cabello (red). (b) Robustness under thermal noise in the loss channel: the environment port is a thermal state with mean photon number $n_{\mathrm{th}}$. Heat maps show $S_3^{\ms M}$ (left) and $S_3^{\ms C}$ (right) as functions of $(\eta,n_{\mathrm{th}})$; dashed lines mark the local bounds ($S_3^{\ms M}=2$ and $S_3^{\ms C}=0$), and above them the local region is shown in blue hues.}
    \label{F-critical-robustness}
\end{figure*}

In the lossless regime, Fig.~\ref{F-MABK-W} shows that increasing the number of parties $N$ enhances the observed nonlocality for both resource families in different ways. For GHZ-type states [Fig.~\hyperref[F-MABK-W]{\ref{F-MABK-W}a}, left], the MABK Bell parameter $S_N^{\ms M}$ increases with $N$ as the effective GKP readout approaches ideal logical Pauli measurements with increasing squeezing. At the same time, the relevant classical benchmark depends on $N$, since any local-hidden-variable (LHV) model satisfies $S_N^{\ms M}\le 2^{\lfloor N/2\rfloor}$. Thus, while larger $N$ yields larger quantum values, the ``difficulty'' of violating the inequality also changes with $N$ through the growing local bound (with odd/even $N$ sharing the same bound when $\lfloor N/2\rfloor$ coincides). For W-type states [Fig.~\hyperref[F-MABK-W]{\ref{F-MABK-W}b}, left], the Cabello functional $S_N^{\ms C}$ likewise increases with $N$, but in this case the local bound is independent of system size, $S_N^{\ms C}\le 0$. Consequently, once the distribution crosses into the nonlocal region, increasing $N$ directly translates into a larger positive gap above the classical threshold, without the moving-target effect present in the MABK family. Overall, both panels confirm the same basic mechanism—higher squeezing makes the binned-homodyne measurement a more faithful approximation of logical Pauli readout--while the $N$-dependence of the classical bound leads to distinct scaling behaviours for GHZ- and W-type witnesses.

Turning to loss [right panels of Fig.~\ref{F-MABK-W}], photon loss suppresses the violations for both resources by blurring the phase-space grid and reducing the distinguishability of neighbouring peaks under quadrature readout. Throughout, we model loss in each mode by a standard beam-splitter channel of transmissivity $\eta$, which mixes the signal with vacuum at the unused port (see Appendix~\hyperref[App:loss-thermal]{C-3} for details). This degradation competes with squeezing: increasing $r_{\mathrm{dB}}$ initially strengthens the correlations, but for $\eta<1$ the loss-induced Gaussian blurring eventually dominates, leading to saturation and a decrease of the Bell value at large squeezing. To quantify performance in a way that is comparable across witnesses with different local bounds, we define the critical squeezing $r_{\mathrm{crit}}$ as the smallest squeezing for which the corresponding inequality is violated. The critical squeezing parameter captures the operational threshold on GKP quality required to certify nonlocality under fixed homodyne-and-binning readout, and it provides a direct basis for comparing robustness: for a fixed transmissivity $\eta$ (and, when relevant, fixed $N$), the resource with the smaller $r_{\mathrm{crit}}$ yields a loophole-friendly violation with less squeezing.

Figure~\ref{F-critical-robustness} makes these thresholds and robustness considerations explicit. In Fig.~\hyperref[F-critical-robustness]{\ref{F-critical-robustness}a} we plot the critical squeezing as a function of the transmissivity $\eta$ for $N=3$, which we take as a representative benchmark as it is the smallest multipartite setting where violations are possible in our measurement model. As expected, the critical squeezing decreases monotonically as $\eta\to 1$, while it grows as loss increases. The comparison between the two resources is captured directly by the ordering of the curves: the W-type witness achieves violation at lower squeezing than the GHZ-type witness. Figure~\hyperref[F-critical-robustness]{\ref{F-critical-robustness}b} further probes robustness against excess noise by replacing the vacuum environment in the loss model with a thermal state of mean photon number $n_{\mathrm{th}}$ (Appendix~\hyperref[App:loss-thermal]{C-3}). The resulting heat maps show the Bell values as functions of $(\eta,n_{\mathrm{th}})$, with dashed contours marking the classical thresholds ($S_3^{\ms M}=2$ for MABK and $S_3^{\ms C}=0$ for Cabello). In both cases, decreasing $\eta$ or increasing $n_{\mathrm{th}}$ drives the system across the contour into the local region, illustrating a clear trade-off: higher transmissivity can compensate for moderate thermal noise, while near-vacuum noise is required to tolerate stronger attenuation. Together, these plots provide an operational ``noise budget'' for homodyne-based Bell tests with finite-energy GKP resources, quantifying both the minimum squeezing needed for a violation and the range of loss and thermal background over which nonlocality remains observable.

Beyond specific inequalities, we verify the nonlocality's strength via the geometric distance to the local polytope, $\mathcal{D}(\v{p})$ \cite{Brito2018} (see Appendix~\ref{App:distance-polytope}). This analysis confirms that W-type correlations are more robust to finite squeezing than GHZ states. Crucially, increasing the measurement settings from $2$ to $3$ drastically amplifies this distance; in the lossless limit, the $3$-input behavior saturates the algebraic maximum even with finite squeezing. This implies that the nonlocality accessible via homodyne detection is significantly richer than what standard two-setting witnesses capture.

\emph{Discussion \& outlook} --  Beyond fault-tolerant quantum computing, GKP encodings are emerging as natural building blocks for photonic quantum technologies, including long-distance quantum communication and repeater/network architectures that rely predominantly on Gaussian operations and homodyne readout~\cite{Noh2019TIT,Fukui2021PRR,Rozpedek2021npjQI,Rozpedek2023PRR}. This makes it timely to understand whether GKP states can also certify Bell nonlocality under similarly experimentally friendly, homodyne-based measurements. We have shown that finitely squeezed GKP resources enable sizeable multipartite Bell violations using only homodyne detection and fixed binning, yielding a simple benchmark compatible with high-efficiency homodyne detection and realistic noise. 

Our main technical result, Eq.~\eqref{Eq:main-overlap}, establishes an explicit and efficiently computable link between finitely squeezed GKP structure and binned-homodyne statistics, that can be extended beyond Bell tests. In particular, it can be extended to prepare-and-measure protocols under homodyne measurements.~\cite{VanHimbeeck2017,Rusca2019SelfTesting,Rusca2020Homodyne,Avesani2021Heterodyne,Wang2023Homodyne,Roch2025,alves2025semideviceindependentrandomnesscertificationdiscretized}. It could provide quantitative tools for privacy certification in distributed quantum sensing, where the central question is to bound what an untrusted observer can infer from accessible data while retaining metrological performance~\cite{shettell2022,ho2024quantumprivatedistributedsensing,Hassani2025,Bugalho2025,dejong2025,junior2025privacycontinuousvariabledistributedquantum,alushi2025privacydistributedquantumsensing}.

Finally, the main experimental challenge is the scalable preparation of multipartite entangled GKP resources. For the GHZ family, comparatively straightforward at the logical level: $\ket{\text{GHZ}_N}$ is stabilizer and can be generated from GKP stabilizer inputs using Clifford operations, which are at the same time implementable with Gaussian couplings and homodyne readout in GKP architectures~\cite{Bourassa2021blueprintscalable,Baranes2023}. By constrant, $\ket{\text{W}_N}$ is nonstabilizer, so preparing it from stabilizer inputs cannot be achieved with Clifford operations alone. Establishing experimentally feasible routes to (W)-type GKP entanglement within (or close to) the Gaussian toolbox therefore remains an interesting open direction, and is closely tied to the need for logical non-Clifford resources, which in the GKP code subspace require additional non-Gaussian ingredients beyond Gaussian protocols \cite{Hahn2025}.

\emph{Acknowledgments--} We gratefully acknowledge financial support from the  EU Horizon Europe (QSNP, grant no. 101114043 \& CLUSTEC, grant no. 101080173 \& ClusterQ, grant no. 101055224), the Simons Foundation (Grant Number 1023171, RC), the Brazilian National Council for Scientific and Technological Development (CNPq, Grants No.307295/2020-6 and No.403181/2024-0), the Financiadora de Estudos e Projetos (grant 1699/24 IIF-FINEP), the Coordenação de Aperfeiçoamento de Pessoal de Nível Superior – Brasil (CAPES) – Finance Code 001, and a guest professorship from the Otto M\o nsted Foundation.  

\emph{Data availability--} All data is publicly available on GitHub at \href{https://github.com/AdeOliveiraJunior/Bell-Nonlocality-from-Gottesman-Kitaev-Preskill-Encoding.git}{Bell Nonlocality from GKP States.}

\bibliography{2-citations}
\appendix
\onecolumngrid

\newpage

\section{A no-go theorem for CHSH violation with Pauli local measurements}
\label{App:nogo-chsh-pauli}

This appendix formalises a simple obstruction to bipartite CHSH violation that arises when the available dichotomic measurements are restricted to the (logical) Pauli observables. The result captures the essential reason why, in settings where the only accessible measurements are implementations of logical $X, Y$ and $Z$ (e.g., via quadrature homodyne detection followed by GKP-aligned binning), CHSH violations cannot be observed..

Let $\mathcal{P} = \{\pm X, \pm Y, \pm Z\}$ denote the set of nontrivial Pauli observables on a qubit, each with eigenvalues $\pm 1$. The CHSH expression for a two-qubit state $\rho$ and dichotomic observables $A_0, A_1$ (Alice) and $B_0, B_1$ (Bob) is
\begin{equation}
    S_{\textrm{CHSH}}(\rho;A_0, A_1, B_0, B_1):= \tr(\rho \mathcal{B}_{\textrm{CHSH}}) \quad \text{with} \quad \mathcal{B}_{\text{CHSH}}:=A_0\otimes (B_0 + B_1) + A_1\otimes(B_0-B_1). 
\end{equation}
Equivalently, $S_{\textrm{CHSH}}(\rho;A_0, A_1, B_0, B_1) = \langle A_0 B_0\rangle+\langle A_0 B_1\rangle+\langle A_1 B_0\rangle-\langle A_1 B_1\rangle$, where $\langle A_i B_j\rangle = \tr(\rho A_i \otimes B_j)$. Any local-hidden-variable model satisfies $|S_\text{CHSH}(\rho;A_0, A_1, B_0, B_1)| \leq 2$.
Before we proceed with the main theorem, we begin with two standard lemmas. The first lemma shows that the CHSH value is invariant under local unitary conjugation, provided the observables are transformed accordingly:
\begin{lem}[Invariance under local conjugation]\label{Lem:invariance-conjugation} For any two-qubit state $\rho$, any unitaries $U_A, U_B$, and any observables $A_i, B_j$,
\begin{equation}
    S_{\text{CHSH}}\qty[(U_A\otimes U_B)\rho (U_A\otimes U_B)^\dagger; A_0, A_1, B_0, B_1] = S_{\text{CHSH}}\qty[\rho; U_A^\dagger A_0U_A, U_A^\dagger A_1 U_A, U_B^\dagger B_0 U_B, U^\dagger_B B_1 U_B].
\end{equation}
\end{lem}
\begin{proof}
    This follows from cyclicity of the trace
    \begin{equation}
        \tr[(U_A\otimes U_B)\rho (U_A\otimes U_B)^\dagger(A_i \otimes B_j)] = \tr[\rho(U_A^\dagger A_iU_A\otimes U_B^{\dagger}B_jU_B)].
    \end{equation}
and linearity in the CHSH combination. 
\end{proof}
Next, the second lemma shows that conjugation by a Clifford unitary preserves the Pauli set:
\begin{lem}[Pauli set is invariant under Clifford conjugation]\label{Lem:app-conjugation} If $U$ is a single-qubit Clifford unitary and $P\in\mathcal{P}$, then $U^\dagger P U \in \mathcal{P}$. 
\end{lem}
\begin{proof}
    This is the defining property of the Clifford group as the normaliser of the single-qubit Pauli Group: $U^\dagger \mathcal{P} U = \mathcal{P}$ up to phases, and for Hermitian Paulis the relevant phase reduces to $\pm 1$.
\end{proof}
We now state the main theorem:
\begin{thm}[No CHSH violation with Pauli measurements on a local-Clifford Bell pair]. Let $\ket{\Phi^+}:=2^{-\tfrac12}(\ket{00}+\ket{11)})$ and let $\rho$ be any two-qubit state of the form:
\begin{equation}\label{Eq:app-modified-bell}
    \rho = (U_A\otimes U_B)\ketbra{\Phi^+}(U_A\otimes U_B)^\dagger,
\end{equation}
where $U_A$ and $U_B$ are single-qubit Clifford unitaries. Assume Alice and Bob are restricted to Pauli observables $A_0, A_1, B_0, B_1 \in \mathcal{P}$. Then,
\begin{equation}
    |S_{\textrm{CHSH}}(\rho;A_0, A_1, B_0, B_1)|\leq 2.
\end{equation}
\end{thm}
\begin{proof}
    Let $\rho$ be as in Eq~\eqref{Eq:app-modified-bell}. Applying Lemma~\ref{Lem:invariance-conjugation} gives
    \begin{equation}\label{Eq:app-reduction}
        S_{\text{CHSH}}(\rho;A_0, A_1, B_0, B_1) = S_{\text{CHSH}}\qty(\ketbra{\Phi^+}; A_0', A_1', B_0', B_1')
    \end{equation}
where $A_i'=U_A^\dagger A_i U_A$ and $B_j' = U_B^\dagger B_j U_B$. By Lemma~\ref{Lem:app-conjugation}, $A_0', A_1', B_0', B_1' \in \mathcal{P}$. Thus it suffices to prove that for any $A_0, A_1, B_0, B_1 \in \mathcal{P}$ one has 
\begin{equation}\label{Eq:app-suffice-proof}
    S_{\textrm{CHSH}}\qty(\ketbra{\Phi^+};A_0, A_1, B_0, B_1) \leq 2.
\end{equation}
We prove this statement with two steps:
\begin{enumerate}
    \item \emph{Step 1: Pauli-Pauli correlator on $\ket{\Phi^+}$.} First, recall the following identity
\end{enumerate}
\begin{equation}\label{Eq:app-identity-Bell}
    \langle \Phi^+|A\otimes B|\Phi^+\rangle = \frac{1}{2}\tr(AB^\top),
\end{equation}
where the transpose is taken in the computational basis. Applying Eq.~\eqref{Eq:app-identity-Bell} to $A\in\{X,Y,Z\}$ and $B\in\{X,Y,Z\}$ and using $X^\top = X, Z^\top = Z, Y^\top = -Y$, and $\tr(PQ) = 2\delta_{P,Q}$ for $P, Q \in \{X, Y, Z\}$ gives
\begin{align}\label{Eq:app-Pauli-correlation-1}
    \langle X \otimes X \rangle_{\Phi^+} &= + 1, \nonumber\\ \langle Y \otimes Y \rangle_{\Phi^+} &= -1, \nonumber\\ \langle Z \otimes Z \rangle_{\Phi^+} &= + 1,
\end{align}
and
\begin{equation}\label{Eq:app-Pauli-correlation-2}
    \langle P\otimes Q\rangle_{\Phi^+} =0 \:\: \text{whenever} \:\: P, Q \in \{X,Y,Z\} \:\: \text{and}\:\: P\neq Q.
\end{equation}
\begin{enumerate}[start=2]
    \item \emph{Step 2: CHSH bound under Pauli choices}. Let $A_0, A_1, B_0, B_1 \in \mathcal{P}$. Define the four correlators $s_{ij}:=\langle A_i \otimes B_j \rangle$. By Eqs.~\eqref{Eq:app-Pauli-correlation-1}-\eqref{Eq:app-Pauli-correlation-2}, each $s_{ij} \in \{0,\pm1\}$ and $s_{ij}\neq0$ can occur only if $A_i$ and $B_j$ are the same Pauli axis (up to a sign). We consider two cases:
\begin{enumerate}
    \item Case 1: $B_1 = \pm B_0$. If $B_1 = B_0$, then
    \begin{equation}
        S_{\text{CHSH}} = s_{00}+s_{10}+s_{01} - s_{11} = 2s_{00},
    \end{equation}
    since $s_{01}=s_{00}$ and $s_{11} = s_{10}$. Hence, $|S_{\text{CHSH}}| = 2|s_{00}| \leq 2$. If $B_1 = -B_0$, then $|S_{\text{CHSH}}| = 2s_{10}$ and again $|S_{\text{CHSH}}| \leq 2.$
\item Case 2: $B_0$ and $B_1$ are different Pauli axes. Then for each fixed $i\in\{0,1\}$, $A_i$ can match at most one of the two distinct axes $\{B_0, B_1\}$ (up to a sign), so at most one of $s_{i0}$ and $s_{i1}$ is nonzero. Consequently,
\begin{align}\label{Eq:app-correlations-3}
    |s_{00}+s_{01}| &\leq 1 ,\nonumber\\|s_{00}-s_{01}| &\leq 1.
\end{align}
Indeed, if both terms in a pair are zero the bound is trivial. If exactly one is zero, then the pair equals $\pm 1$. They cannot both be nonzero because that would require $A_i$ to coincide with both distinct Pauli axes. Using the triangle inequality and Eq.~\eqref{Eq:app-correlations-3},
\begin{equation}
    |S_{\text{CHSH}}| = |s_{i0}+s_{i1}|\le 1 \quad \text{and} \quad |s_{i0}-s_{i1}|\le 1 \:\:\text{for} \:\: i=0,1.
\end{equation}
\end{enumerate}
\end{enumerate}
Thus $|S_{\text{CHSH}}|\leq 2$ holds in both cases, proving that Eq.~\eqref{Eq:app-suffice-proof}. Combining with the reduction given by Eq.~\eqref{Eq:app-reduction} completes the proof of the theorem.
\end{proof}
Observe that the reduction from $\rho$ in Eq~\eqref{Eq:app-reduction} to $\ket{\Phi^+}$ uses the fact that Clifford conjugation preserves the Pauli set: $U^\dagger P U \in \{X,Y,Z\}$ whenever $U$ is a Clifford and $P\in \mathcal{P}$. This is false for general single-qubit unitaries, hence the assumption in Eq.~\eqref{Eq:app-reduction} is essential for the theorem as stated. Finally, in GKP-based implementations, logical $X$, $Y$, and $Z$ measurements are commonly realized by homodyne measurement of a quadrature (or a rotated quadrature), followed by modular binning into two outcomes aligned with the GKP lattice. When this procedure implements dichotomic observables that are well-approximated by the logical Pauli set, the above theorem explains the absence of CHSH violations in the bipartite case: CHSH violation requires measurement directions outside the Pauli axes (e.g.\ $\propto X\pm Z$), which are not available under a Pauli restriction.

\section{Multipartite Bell inequalities}\label{App:bell-inequalities}

In this Appendix we define the multipartite Bell functionals used in the main text for the GHZ and $W$ families. Throughout, each party $j\in\{1,\dots,N\}$ receives an input $x_j\in\{0,1\}$ selecting one of two binned-homodyne POVMs $\{M_{0}^{(x_j)},M_{1}^{(x_j)}\}$, and returns an output bit $o_j\in\{0,1\}$, yielding the behaviour $p(\v o|\v x)$. Given a two-outcome POVM $\{M_{0}^{(x)},M_{1}^{(x)}\}$ we associate the corresponding dichotomic observable
\begin{equation}\label{Eq:dichotomic-observable}
    A^{(x)} := \sum_{o\in\{0,1\}} (-1)^o\, M_{o}^{(x)} = M_{0}^{(x)}-M_{1}^{(x)} ,
\end{equation}
whose outcome is the $\pm1$ variable $a:=(-1)^o$. For a setting string $\v x\in\{0,1\}^N$, the full correlator is
\begin{equation}\label{Eq:full-correlator}
    E(\v x):=\Bigl\langle \prod_{j=1}^N a_j\Bigr\rangle =\sum_{\v o\in\{0,1\}^N} (-1)^{\sum_{j} o_j}\, p(\v o|\v x) =\tr\left[\rho\,\bigotimes_{j=1}^N A_j^{(x_j)}\right],
\end{equation}
where $A_j^{(x_j)}$ denotes the observable~\eqref{Eq:dichotomic-observable} acting on mode $j$. In our implementation, the settings correspond to homodyne angles that (in the ideal GKP limit) realise logical Pauli measurements, e.g., $(Y,X)$ or $(Z,X)$ as specified below.

\subsection{MABK inequality for GHZ-type resources}\label{App:MABK}

To witness GHZ-type full-correlation nonlocality we use the Mermin--Ardehali--Belinskii--Klyshko (MABK) family. Here each party chooses between two dichotomic observables $A_j := A_{j}^{(0)}, A'_j := A_{j}^{(1)}$, which in the main text correspond to the effective logical settings $(Y,X)$, i.e., $x=0$ is the binned-homodyne measurement at $\theta_Y$ and $x=1$ at $\theta_X$. Define the Bell operators $B_N$ and $B'_N$ recursively by $B_1:=A_1$ and $B_1':=A_1'$, and,
for $N\ge2$,
\begin{subequations}\label{Eq:MABK-recursion}
\begin{align}
B_N &:= \frac{1}{2}\Bigl[B_{N-1}\otimes\left(A_N + A_N'\right)
+ B_{N-1}'\otimes\left(A_N - A_N'\right)\Bigr],\\[2pt]
B_N'&:= \frac{1}{2}\Bigl[B_{N-1}'\otimes\left(A_N + A_N'\right)
- B_{N-1}\otimes\left(A_N - A_N'\right)\Bigr].
\end{align}
\end{subequations}
The corresponding (dimensionless) Bell parameter is $S_N^{\ms M} := 2^{\lfloor N/2\rfloor}\,\max\left\{ \bigl|\langle B_N\rangle\bigr|,\ \bigl|\langle B_N'\rangle\bigr| \right\}$, with $\langle B_N\rangle=\tr(\rho\,B_N)$ (and similarly for $B_N'$).
Any local-hidden-variable (LHV) model obeys $S_N^{\ms M}\le 2^{\lfloor N/2\rfloor}$, so that $S_N^{\ms M}>2^{\lfloor N/2\rfloor}$ certifies multipartite Bell nonlocality.

\subsection{Cabello-type inequality for $W$-type resources}\label{App:Cabello}

To probe nonlocality associated with $W$-type entanglement we use a Cabello/Hardy-type probability inequality \cite{Cabello2002,Heaney2011}, which only involves joint probabilities for a restricted set of measurement configurations. Here each party chooses between two binned-homodyne measurements corresponding to the effective logical settings $(Z,X)$, i.e.\ $x=0$ is the setting at $\theta_Z$ and $x=1$ at $\theta_X$.

Let $\v 0:=(0,\dots,0)$ and $\v 1:=(1,\dots,1)$, and let $\v e_j$ denote the $N$-bit string with a single $1$ at site $j$ and zeros elsewhere. For each unordered pair $(i,j)$ with $i<j$, define the setting string $\v x^{(i,j)}$ by $(x^{(i,j)})_i=(x^{(i,j)})_j=1$ and $(x^{(i,j)})_k=0$ for $k\neq i,j$ (two $X$'s and the rest $Z$'s). We then define
\begin{equation}\label{Eq:Cabello-functional}
\begin{aligned}
S_N^{\ms C} := p(\v 0\mid \v 0) + \sum_{j=1}^N p(\v e_j\mid \v 0)
 - \sum_{1\le i<j\le N}\Bigl[p(\v e_i\mid \v x^{(i,j)})+p(\v e_j\mid \v x^{(i,j)})\Bigr]-\,p(\v 0\mid \v 1) - p(\v 1\mid \v 1).
\end{aligned}
\end{equation}
Any LHV model satisfies the bound $S_N^{\ms C}\le 0,$ so that $S_N^{\ms C}>0$ certifies multipartite Bell nonlocality. As emphasized in the main text, \eqref{Eq:Cabello-functional} depends only on probabilities from the three classes of setting strings: all-$Z$ ($\v x=\v 0$), all-$X$ ($\v x=\v 1$), and two-$X$ settings $\v x^{(i,j)}$, totaling $2+\binom{N}{2}$ measurement configurations.

\subsection{Distance to the polytope as a linear program}\label{App:distance-polytope}
Beyond specific inequalities one can also verify nonlocality via the distance to the local polytope. Here we establish the distance from a distribution $\v{p}$ to the local polytope $\mathcal{L}$ of classical correlations as a linear program.  
For an $N$-partite Bell scenario, each party having  $I$ inputs and $O$ outputs, any classical distribution $q(\v{o}|\v{b})$ must be decomposable as
\begin{equation}
    q(\v{o}|\v{b}) = \sum_{\lambda}q_\lambda\prod_{j=1}^N q(o_j|b_j), \quad \text{for } \v{o}\in\{1,\dots,O\}^N, \: \v{b}\in\{1,\dots,I\}^N ,\label{eq: LHV_model_Nparties}
\end{equation}
where $q(o_j|b_j)$ corresponds to the local distribution observed by party $j$ and $q(\lambda)$ is the probability distribution of the common hidden variable $\lambda$. By denoting $\v{p}\in \mathbb{R}^{(OI)^N}$ the real vector with entries $p_{i=(o_1,\dots,o_N,b_1,\dots, b_N)} = p(\v{o}|\v{b})$,  the set of all classical distributions $\v{q}$ given by Eq.~(\ref{eq: LHV_model_Nparties}) defines the classical polytope $\mathcal{L}_N$. Linear programming allows to quantify the distance $\mathcal{D}(\v{p})$ of any distribution $\v{p}$ from the local polytope.  The distance can be defined as the following minimization 
\begin{equation}
\begin{aligned}
    \mathcal{D}(\v{p}) = \min &\|\v p-\v q\|_1 \\
    \text{s.t.} \quad & \v{q} =\sum_{\lambda = 0}^{n-1}q_\lambda\v{q}^\lambda, \quad
    q_i^\lambda = \prod_{j=1}^N\delta_{f_j(b_j,\lambda),o_j}, \quad
    \sum_{\lambda=0}^{n-1} q_\lambda = 1, \quad q_\lambda\geq0.
\end{aligned}
\end{equation}
where we have denoted by $f_j$ the finite deterministic response functions of party $j$. Note that, since the number of inputs and outputs is finite, there is also a finite number  $n$ of deterministic strategies. Also note that the function $\mathcal{D}(\v{p})$ can be stated as a linear program because all the constraints are linear and the objective function, the 1-norm, can be linearized through slack variables \cite{skrzypczyk2023semidefinite}.
Any classical distribution must have distance zero, while any positive distance indicates that the target distribution is  unfactorizable as a LHV model. 

Figure~\ref{F-LP} compares the geometric distance to the local polytope for distributions generated by GHZ and W states in 2 scenarios: the standard 2-setting case ($\theta_X$,$\theta_Y$ or $\theta_X$,$\theta_Z$) and an extended 3-setting case ($\theta_X$,$\theta_Y$,$\theta_Z$). Consistent with the results in the main text, the $W$-state distributions (dashed lines) exhibit nonlocality at lower squeezing thresholds than the $GHZ$ distributions (solid lines). Crucially, increasing the number of settings from two to three yields a dramatic increase in the distance $\mathcal{D}(\v{p})$ of the $W$-state distributions. In the lossless limit ($\eta=1$), the distributions saturate the algebraic maximum of $\mathcal{D}(\v{p}) = 2$, demonstrating maximal distinguishability (zero overlap) with the local set. The noiseless correlations reside on the boundary of the no-signaling polytope.
\begin{figure*}
    \centering
    \includegraphics{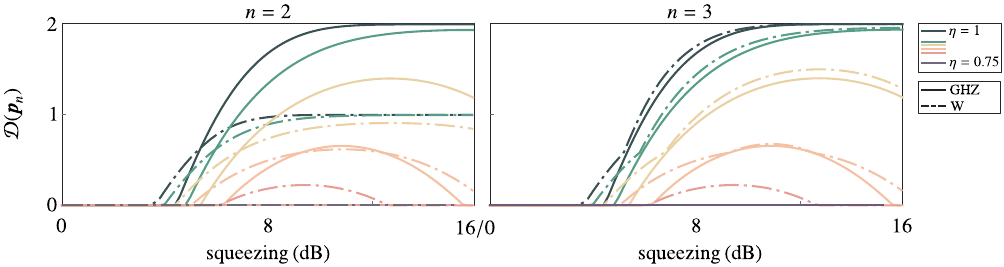}
    \caption{\emph{Nonlocality beyond Bell.}  Distance to the local polytope, $\mathcal{D}(\v p_n)$, for the binned-homodyne behaviours generated by finitely squeezed GKP-encoded GHZ (solid lines) and $W$ (dashed lines) resources, as a function of the squeezing (dB). Left: $n=2$ inputs, corresponding to the effective Pauli settings $(X,Y)$ for GHZ and $(X,Z)$ for $W$. Right: $n=3$ inputs, corresponding to $(X,Y,Z)$. Colours indicate the channel transmissivity $\eta$ (from $\eta=1$ to $\eta=0.75$). Positive values of $\mathcal{D}$ certify Bell nonlocality.}
    \label{F-LP}
\end{figure*}
\section{Finite-energy GKP states\label{App-finite-energy-GKP-state}}

In this appendix, we derive the finite-energy wave functions and the related quantities which follows from it for a finite-energy GKP state. 

\subsection{Energy-filter matrix element}

The central piece for describing the finite-energy GKP wavefunctions is to compute how $e^{-\epsilon \hat{n}}$ acts on a position eigenstate $\ket{x_0}_x$. To do so, we begin by computing the matrix element of the energy filter $\langle x | e^{-\epsilon \hat{n}} | x_0 \rangle$ for $\epsilon>0$. Working with dimensionless quadratures, the number operator can be written as $\hat{n} = \hat{H}-\tfrac12$, which implies $e^{-\epsilon\hat{n}} = e^{-\tfrac{\epsilon}{2}}e^{-\epsilon \hat{H}}$. Consequently, one finds~\cite{Sakurai2017}: $\langle x | e^{-\epsilon \hat{n}} | x_0 \rangle = e^{-\tfrac{\epsilon}{2}} K(x, x_0; \epsilon)$, where $K(x, x_0; \epsilon) \equiv \langle x | e^{-\epsilon \hat{H}} | x_0 \rangle$ denotes the Euclidean (imaginary-time) propagator of the harmonic oscillator. The propagator can be obtained via an eigenfunction expansion. The normalised energy eigenfunctions are
\begin{equation}
    \psi_n(x) = \frac{1}{\pi^{\tfrac14}\sqrt{2^n n!}} e^{-\frac{x^2}{2}} H_n(x)\quad \text{with} \quad n = 0, 1, 2, \dots,
\end{equation}
where $H_n$ denote the Hermite polynomials. The corresponding eigenvalues are $E_n = n+\tfrac12$. Since the functions $\psi_n$ are real, the propagator can be written in the following form:
\begin{equation}\label{Eq:app-kernel-1-sum}
    K(x, x_0; \epsilon) = \sum_{n=0}^{\infty} e^{-\epsilon E_n} \psi_n(x) \psi_n(x_0) = \frac{e^{-\tfrac{\epsilon}{2}}e^{-\frac{x^2+x_0^2}{2}}}{\sqrt{\pi}}\sum_{n=0}^{\infty} \frac{(e^{-\epsilon})^n }{2^n n!} H_n(x)H_n(x_0).
\end{equation}
Using the Mehler formula~\cite{rainville1960special}, the sum appearing in Eq.~\eqref{Eq:app-kernel-1-sum} admits a closed-form expression:
\begin{equation}\label{Eq:app-mehler-formula}
    \sum_{n=0}^{\infty} \frac{H_n(x) H_n(x_0)}{2^n n!} \ell^n = 
\frac{1}{\sqrt{1 - \ell^2}} \exp \qty( \frac{2x x_0 \ell - (x^2 + x_0^2) \ell^2}{1 - \ell^2}),
\end{equation}
with $\ell=e^{-\epsilon}$. Then, substituting Eq.~\eqref{Eq:app-mehler-formula} into Eq.~\eqref{Eq:app-kernel-1-sum}, the kernel takes the form
\begin{equation}
    K(x,x_0;\epsilon) = \frac{1}{\sqrt{\pi}}\frac{\exp(-\tfrac{\epsilon}{2})}{\sqrt{1 - \ell^2}}\exp(-\frac{x^2+x_0^2}{2}) \exp \qty( \frac{2x x_0 \ell - (x^2 + x_0^2) \ell^2}{1 - \ell^2}) =  \frac{1}{\sqrt{\pi}}\frac{\exp(-\tfrac{\epsilon}{2})}{\sqrt{1 - \ell^2}}\exp\qty[ \frac{2x x_0 \ell - (x^2 + x_0^2) (1+\ell^2)}{1 - \ell^2}]. 
\end{equation}
Introducing the following hyperbolic identities,
\begin{align}
     1-\ell^2 &= 2e^{-\epsilon}\sinh \epsilon, \\
     \frac{1+\ell^2}{1-\ell^2} &=\coth \epsilon, \\
     \frac{2\ell}{1-\ell^2} &= \frac{1}{\sinh\epsilon},
\end{align}
the propagator can be simplified as follows:
\begin{equation}
    K(x,x_0;\epsilon)  = \frac{1}{\sqrt{2\pi \sinh\epsilon}}\exp\qty[-\frac{(x^2 + x_0^2)\cosh\epsilon - 2x x_0}{2\sinh\epsilon}].
\end{equation}
Finally, the matrix element of the energy filter reads
\begin{equation}\label{Eq:app-matrix-element}
    \langle x|e^{-\epsilon \hat{n}}|x_0\rangle  = \frac{e^{\tfrac{\epsilon}{2}}}{\sqrt{2\pi \sinh\epsilon}}\exp\qty[-\frac{(x^2 + x_0^2)\cosh\epsilon - 2x x_0}{2\sinh\epsilon}].
\end{equation}
As a sanity check, note that in the limit $\epsilon\to0^+$, Eq.~\eqref{Eq:app-matrix-element} tends to $\delta(x-x_0)$. In the opposite limit, $\epsilon\to\infty$, one finds $K(x,x_0;\epsilon)\to e^{-\epsilon/2}\psi_0(x)\psi_0(x)$ and $\langle x|e^{-\epsilon \hat{n}}|x_0\rangle \to \psi_0(x)\psi_0(x)$, corresponding to a projection onto the vacuum.

To reveal the structure relevant for GKP states, we complete the square in the exponent of Eq.~\eqref{Eq:app-matrix-element}:
\begin{align}\label{Eq:app-complete-square}
    -\frac{(x^2 + x_0^2)\cosh\epsilon - 2x x_0}{2\sinh\epsilon} &= -\frac{\cosh\epsilon}{2\sinh \epsilon}(x-x_0\sech\epsilon)^2-\frac{\sinh\epsilon}{2\cosh\epsilon}x_0^2 \nonumber\\
    &=\frac{1}{4\sigma^2}(x-x_0 \sech\epsilon)^2 -\frac{x_0^2}{4\Sigma^2},
\end{align}
where we have defined
\begin{align}
    \sigma^2&:=\frac{1}{2}\tanh\epsilon, \\
    \Sigma^2&:=\frac{1}{2}\coth \epsilon = \frac{1}{4\sigma^2}.
\end{align}
The parameters $\sigma^2$ and $\Sigma^2$ do not directly correspond to the variances of the Gaussians. Instead, the peak variance is $2\sigma^2$, while the envelope variance is $2\Sigma^2$. The factor $\sech\epsilon$ describes a contraction of the underlying lattice.

The ideal square-lattice GKP logical states are Dirac combs $\ket{\mu_L^{\text{ideal}}} \propto \sum_{s\in \mathbbm{Z}} \ket{x=(2s+\mu)\sqrt{\pi}}$. Applying the energy filter yields finite-energy states $\ket{\mu_L^{\epsilon}} \propto \ket{\mu_L^{\text{ideal}}}$. Their position-space wavefunctions are obtained via the integral transform with kernel given by Eq.~\eqref{Eq:app-matrix-element}, which, for Dirac combs, collapses to a discrete sum. Using Eq.~\eqref{Eq:app-complete-square}, the resulting finite-energy codewords are normalisable and can be specified by their position-space wavefunctions $\psi_{\mu_L^\epsilon}(x):={}_x\langle{x\,|\,\mu_L^\epsilon\rangle}$, with
\begin{subequations}\label{Eq:GK-wave-f}
\begin{align}
\psi_{0_L^\epsilon}(x)
&= \frac{1}{\sqrt{\mathcal{N}_0}} \sum_{s\in\mathbb{Z}}
\exp\qty[{-\frac{(x-2s\sqrt{\pi}\sech\epsilon)^2}{4\sigma^2}}-\frac{(2s\sqrt{\pi})^2}{4\Sigma^2}],\\[4pt]
\psi_{1_L^\epsilon}(x)
&= \frac{1}{\sqrt{\mathcal{N}_1}} \sum_{s\in\mathbb{Z}}
\exp\qty[-\frac{[x-(2s+1)\sqrt{\pi}\sech\epsilon]^2}{4\sigma^2} -\frac{[(2s+1)\sqrt{\pi}]^2}{4\Sigma^2}].
\end{align}
\end{subequations}

\subsection{Wigner functions of the GKP Pauli operators}
We represent oscillator states in phase space via the Wigner transform
\begin{equation}
W_A(q,p) :=\frac{1}{2\pi}\int_{-\infty}^{\infty}dy e^{ipy} \Big\langle q-\frac{y}{2}\Big|A\Big|q+\frac{y}{2}\Big\rangle,
\end{equation}
with $(q,p)$ the canonical quadratures. For the energy-filter model $e^{-\epsilon\hat n}$, the Wigner functions of the operators
$\sigma_k^\epsilon$ admit a convenient closed form as a Gaussian mixture on a contracted square lattice:
\begin{equation}\label{eq:GKP_Pauli_Wigner}
W_{\sigma_k^\epsilon}(q,p) = (-1)^{\delta_{k,2}} \sum_{\v{m}\in\mathcal M_k} c_{\v{m}}^\epsilon\, s_k(\v{m})\, G_{\v{\mu}_{\v{m}}^\epsilon,\v\Sigma_W^\epsilon}(q,p), \qquad k=0,1,2,3.
\end{equation}
Here $\v x:=(q,p)^T$, and
\begin{equation}\label{Eq:app-Gaussian-kernel}
G_{\v{\mu},\v\Sigma_W}(\v x) :=\frac{1}{2\pi\sqrt{\det\v\Sigma_W}} \exp\left[-\frac{1}{2}(\v x-\v{\mu})^T \v\Sigma_W^{-1}(\v x-\v{\mu})\right]
\end{equation}
is a centred two-dimensional Gaussian with mean $\v{\mu}$ and covariance matrix
$\v\Sigma_W$.
For the finite-energy GKP states of Eq.~\eqref{Eq:GK-wave-f}, the covariance is isotropic and equals
\begin{equation}\label{Eq:app-covariance-isotropic}
\v\Sigma_W^\epsilon=\sigma^2\,\iden=\frac{1}{2}\tanh(\epsilon)\,\iden,
\end{equation}
while the phase-space means are located on a contracted square lattice,
\begin{equation}\label{Eq:contracted-square-lattice}
\v{\mu}_{\v{m}}^\epsilon
=\sech(\epsilon)\frac{\sqrt{\pi}}{2}\v m \:\: \text{and} \:\: 
\v{m}=\begin{bmatrix}m_1\\m_2\end{bmatrix}\in\mathbb Z^2.
\end{equation}
The sets $\mathcal M_k$ select one of four lattice cosets:
\begin{subequations}\label{Eq:app-MK-sets}
\begin{align}
\mathcal M_0 &= \{\v m\in\mathbb Z^2 : m_1\equiv 0 \ (\mathrm{mod}\ 2),\; m_2\equiv 0 \ (\mathrm{mod}\ 2)\},\\
\mathcal M_1 &= \{\v m\in\mathbb Z^2 : m_1\equiv 1 \ (\mathrm{mod}\ 2),\; m_2\equiv 0 \ (\mathrm{mod}\ 2)\},\\
\mathcal M_2 &= \{\v m\in\mathbb Z^2 : m_1\equiv 1 \ (\mathrm{mod}\ 2),\; m_2\equiv 1 \ (\mathrm{mod}\ 2)\},\\
\mathcal M_3 &= \{\v m\in\mathbb Z^2 : m_1\equiv 0 \ (\mathrm{mod}\ 2),\; m_2\equiv 1 \ (\mathrm{mod}\ 2)\}.
\end{align}
\end{subequations}
and the sign pattern is
\begin{subequations}\label{Eq:app-sign-patterns}
\begin{align}
s_0(\v m) &= 1,\\
s_1(\v m) &= (-1)^{m_2/2},\\
s_2(\v m) &= (-1)^{(m_1+m_2)/2},\\
s_3(\v m) &= (-1)^{m_1/2}.
\end{align}
\end{subequations}
which is responsible for the alternating positive and negative peaks of the Wigner function.
Finally, the envelope weights are
\begin{equation}\label{Eq:app-cm-weights}
c_{\v{m}}^\epsilon
=\exp\left[-\frac{\pi}{4}\tanh(\epsilon)\big(m_1^2+m_2^2\big)\right],
\end{equation}
guaranteeing convergence of the lattice sum.

\subsection{Loss and thermal-noise model}
\label{App:loss-thermal}

In the main text we model attenuation and Gaussian noise by coupling each signal mode to an environment mode on a beam splitter of transmissivity $\eta\in[0,1]$. The environment is taken to be either vacuum (pure loss) or a thermal state (thermal loss). 
The thermal Wigner function is an isotropic Gaussian, $W_{\tau_{n_{\mathrm{th}}}}(q,p)=G_{\v 0,\v\Sigma_{\mathrm{th}}}(q,p)$, with $\v\Sigma_{\mathrm{th}}=\left(n_{\mathrm{th}}+\tfrac12\right)\iden$, where $n_{\mathrm{th}}$ is the thermal mean photon number $n_{\mathrm{th}}$. Consequently, the vacuum corresponds to $n_{\mathrm{th}}=0$ and $\v\Sigma_{\mathrm{th}}=\tfrac12\iden$. A key simplification for our purposes is that the the noisy-loss channel is  a Gaussian channel~\cite{Weedbrook2012}, hence it maps each Gaussian component of a Wigner function to another Gaussian. More precisely, for the Gaussian kernel of Eq.~\eqref{Eq:app-Gaussian-kernel} one has the transformation
\begin{equation}\label{Eq:App-gaussian-kernel-channel}
G_{\v\mu,\v\Sigma}\ \longmapsto\
G_{\sqrt{\eta}\,\v\mu,\ \eta\,\v\Sigma+(1-\eta)\v\Sigma_{\mathrm{th}} }.
\end{equation}
Equation~\eqref{Eq:App-gaussian-kernel-channel} can be understood directly at the level of first and second moments: in the Heisenberg picture the quadrature vector $\hat{\v x}=(\hat q,\hat p)^T$ transforms as $\hat{\v x}\mapsto \sqrt{\eta}\,\hat{\v x}+\sqrt{1-\eta}\,\hat{\v x}_E$, so that the mean rescales by $\sqrt{\eta}$ and the covariance picks up an additive isotropic term weighted by $(1-\eta)$. In particular, the excess-noise contribution is proportional to $(1-\eta)n_{\mathrm{th}}$, as expected for noise injected through the unused port of the beam splitter. Applying \eqref{Eq:App-gaussian-kernel-channel} term-wise to the Gaussian-mixture representation \eqref{eq:GKP_Pauli_Wigner} shows that the mixture structure of the finitely-squeezed GKP Pauli operators is preserved under loss and thermal noise. Defining
\begin{equation}\label{Eq:App-sigma-tilde}
\tilde{\v\Sigma}_W^{\,\epsilon}(\eta,n_{\mathrm{th}})
:=
\eta\,\v\Sigma_W^\epsilon+(1-\eta)\v\Sigma_{\mathrm{th}}
=
\tilde{\sigma}^2(\eta,n_{\mathrm{th}},\epsilon)\,\iden,
\:\:\text{with}\:\:
\tilde{\sigma}^2(\eta,n_{\mathrm{th}},\epsilon)
=
\eta\,\sigma^2+(1-\eta)\left(n_{\mathrm{th}}+\tfrac12\right),
\end{equation}
and the rescaled lattice means $\tilde{\v\mu}_{\v m}^{\,\epsilon}(\eta) := \sqrt{\eta}\,\v\mu_{\v m}^\epsilon = \sqrt{\eta}\,\sech(\epsilon)\frac{\sqrt{\pi}}{2}\v m$, we obtain, for $k=0,1,2,3$, the following Wigner function
\begin{equation}\label{Eq:App-GKP-Pauli-Wigner-loss}
W_{{\eta,n_{\mathrm{th}}}(\sigma_k^\epsilon)}(q,p) = (-1)^{\delta_{k,2}} \sum_{\v m\in\mathcal M_k} c_{\v m}^\epsilon\, s_k(\v m)\, G_{\tilde{\v\mu}_{\v m}^{\,\epsilon}(\eta),\,\tilde{\v\Sigma}_W^{\,\epsilon}(\eta,n_{\mathrm{th}})}(q,p).
\end{equation}
For pure loss (vacuum environment) one simply sets $n_{\mathrm{th}}=0$, giving $\tilde{\sigma}^2(\eta,0,\epsilon)=\eta\sigma^2+\tfrac{(1-\eta)}{2}$.

\section{Single-mode overlap}

Let $M_o^{(b)}$ be the single-mode POVM element corresponding to setting $b$ and binned outcome $o$ and let the $N$-mode measurement operator for settings $\v b = (b_1, ..., b_N)$ and outcomes $\v o = (o_1, ..., o_N)$ be the product POVM $M_{\v o}^{(\v b)}:= M_{o_1}^{(b_1)}\otimes \cdots \otimes M_{o_N}^{(b_N)}$. For any physical $N$-mode state $\rho$, Born's rule gives $p(\v o |\v b) = \tr(\rho M_{\v o}^{\v b})$. For any physical $N$-mode state $\rho$ that is an encoded qubit state (i.e., $\rho = V_{\epsilon}^{\otimes N} \rho_{\text{qubit}} V_{\epsilon}^{\dagger\otimes N}$ for some $N$-qubit state $\rho_{\text{qubit}}$), it admits a decomposition as linear combination of tensor-product operators built from the single-mode finite energy GKP Paulis $\sigma_k^\epsilon$: $\rho = \frac{1}{\mathcal{N}}\sum_{k\in\{0,1,2,3\}^N} c_k \sigma^{\epsilon}_{k_1}\otimes \cdots \otimes \sigma^{\epsilon}_{k_N}$, where $c_k \in \mathbbm{R}$ and $\mathcal{N}$ normalises the state. Hence, one can factorises the Born's rule as
\begin{align}
    p(\v o |\v b)&= \frac{1}{\mathcal{N}}\sum_k c_k \tr[(\sigma^{\epsilon}_{k_1}\otimes \cdots \otimes \sigma^{\epsilon}_{k_N})(M_{o_1}^{(b_1)}\otimes \cdots \otimes M_{o_N}^{(b_N)})] = \frac{1}{\mathcal{N}}\sum_k c_k \prod_{j=1}^N \tr(M_{o_j}^{(b_j)} \sigma_{k_j}^{\epsilon}) \nonumber\\ &=\frac{1}{\mathcal{N}}\sum_k c_k \prod_{j=1}^N t_{k_j,o_j}^{(b_j)}.
\end{align}
where $t_{k,o}^{(b)}:=\tr(M_{o_j}^{(b_j)} \sigma_{k_j}^{\epsilon})$ is the single-mode overlap.

Let $\hat{q}_\theta:=\hat q \cos\theta + \hat p \sin \theta$ denote the rotated quadrature, with generalised eigenstate $\ket{q_\theta}$. Since a setting $b$ fixes an angle $\theta_b$, we write $M^{(b)}_o := M^{(\theta_b)}_o = \int_{R_o} \textrm{d}q_\theta \ketbra{q_\theta}$, where the binning sets are the periodic union of intervals
\begin{equation}
    R_o= \bigcup_{s\in\mathbb Z} \qty[ (2s+o)\sqrt{\pi}-\frac{\sqrt{\pi}}{2}, (2s+o)\sqrt{\pi}+\frac{\sqrt{\pi}}{2}].
\end{equation}
Now, observe that the problem of computing the single-mode overlap reduces to computing a matrix element and integrating over bins. One can directly see this by linearity of the trace and the definition of $M_o^{(\theta)}$, i.e.,
\begin{equation}
    t^{(\theta)}_{k,o}=\tr(\int_{R_o} \textrm{d}q_\theta \ketbra{q_\theta}\sigma_k^\epsilon) = \int_{R_o}\textrm{d}q_\theta \langle q_\theta|\sigma_k^\epsilon|q_\theta\rangle.
\end{equation}
The next Lemma expresses the diagonal element $\langle q_\theta|A|q_\theta\rangle$ as a marginal of the Wigner function
\begin{lem}[Wigner marginal at $\theta=0$]\label{Lem:app-wigner-marginal}
Let $A$ be a trace-class operator on a single bosonic mode, and define its Wigner function by
\begin{equation}\label{Eq:app-Wigner-def}
    W_A(q,p) :=\frac{1}{2\pi}\int_{-\infty}^{\infty}\mathrm{d}y\, e^{ipy}
    \Big\langle q-\frac{y}{2}\Big|A\Big|q+\frac{y}{2}\Big\rangle.
\end{equation}
Then for all $q\in\mathbb{R}$,
\begin{equation}\label{Eq:wigner-marginal-zero}
    \int_{-\infty}^{\infty}\mathrm{d}p\, W_A(q,p) = \langle q|A|q\rangle.
\end{equation}
\end{lem}
\begin{proof}
Integrating \eqref{Eq:app-Wigner-def} over $p$ and using $\int_{-\infty}^{\infty}dp\,e^{ipy}=2\pi\delta(y)$ gives
\begin{align}
\int_{-\infty}^{\infty} \textrm{d}p W_A(q,p)
&=\int_{-\infty}^{\infty} \textrm{d}p\frac{1}{2\pi}\int_{-\infty}^{\infty} \textrm{d}y e^{ipy}
\Big\langle q-\frac{y}{2}\Big|A\Big|q+\frac{y}{2}\Big\rangle =\frac{1}{2\pi}\int_{-\infty}^{\infty} dy(2\pi)\delta(y)
\Big\langle q-\frac{y}{2}\Big|A\Big|q+\frac{y}{2}\Big\rangle =\langle q|A|q\rangle.
\end{align}
\end{proof}
Lemma~\ref{Lem:app-wigner-marginal} treats the position basis $(\theta=0)$, while our measurement protocol uses homodyne angles $\theta_b$. The next corollary extends the marginal identity to these rotate quadratures, allowing us to express $\langle q_\theta|A|q_\theta\rangle$ in terms of $W_A(q,p)$ without redefining the Wigner function.

\begin{cor}[Rotated Wigner marginal]\label{Cor:app-rotated-wigner-marginal}
Let $U_\theta:=e^{-i\theta\hat n}$ be the phase-shift unitary. In the Heisenberg picture it acts as
\begin{align}\label{Eq:app-quad-rotation}
    U_\theta^\dagger\hat qU_\theta &= \hat q\cos\theta+\hat p\sin\theta = \hat q_\theta,
\\
    U_\theta^\dagger\hat pU_\theta &= -\hat q\sin\theta+\hat p\cos\theta = \hat p_\theta.
\end{align}
Then for any trace-class operator $A$ and all $q_\theta\in\mathbb{R}$,
\begin{equation}\label{Eq:rotated-marginal}
\langle q_\theta|A|q_\theta\rangle =
\int_{-\infty}^{\infty}\mathrm{d}p_\theta
W_A\left(q_\theta\cos\theta - p_\theta\sin\theta, q_\theta\sin\theta + p_\theta\cos\theta\right).
\end{equation}
\end{cor}
\begin{proof}
Since $|q_\theta\rangle=U_\theta^\dagger|q\rangle$, we have $\langle q_\theta|A|q_\theta\rangle=\langle q|U_\theta A U_\theta^\dagger|q\rangle$. Applying Lemma~\ref{Lem:app-wigner-marginal} to the operator $U_\theta A U_\theta^\dagger$ gives
\begin{equation}\label{Eq:marginal-B}
\langle q_\theta|A|q_\theta\rangle
=
\int_{-\infty}^{\infty}\mathrm{d}p W_{U_\theta A U_\theta^\dagger}(q_\theta,p).
\end{equation}
Now use covariance of the Wigner function under the symplectic transformation induced by \eqref{Eq:app-quad-rotation}. Writing $\v x=(q,p)^\top$ and introducing
\begin{equation}
S(\theta):=
\begin{pmatrix}
\cos\theta & \sin\theta\\
-\sin\theta & \cos\theta
\end{pmatrix},
\end{equation}
the Heisenberg action is $U_\theta^\dagger\hat{\v x}U_\theta=S(\theta)\hat{\v x}$, and Wigner covariance reads $W_{U_\theta A U_\theta^\dagger}(\v x)=W_A\left[S(\theta)^{-1}\v x\right]$. Thus, we obtain $W_{U_\theta A U_\theta^\dagger}(q,p)=
W_A\left(q\cos\theta - p\sin\theta, q\sin\theta + p\cos\theta\right)$. Substituting this into \eqref{Eq:marginal-B} and renaming the integration variable $p\mapsto p_\theta$ gives \eqref{Eq:rotated-marginal}.
\end{proof}
We now have all the elements needed to state a theorem which provides a series representation for the single-mode overlap. 
\begin{thm}[Single-mode overlap for binned homodyne on finite-energy GKP Paulis] Let $\sigma_k^{\epsilon}$ denote the finite-energy single-mode GKP Pauli operators whose Wigner functions admit the Gaussian-mixture representation
\begin{equation}\label{Eq:mix}
W_{\sigma_k^\epsilon}(q,p) = (-1)^{\delta_{k,2}} \sum_{\v{m}\in\mathcal M_k} c_{\v{m}}^\epsilon\, s_k(\v{m})\, G_{\v{\mu}_{\v{m}}^\epsilon,\v\Sigma_W^\epsilon}(q,p), \qquad k=0,1,2,3.
\end{equation}
with $\Sigma_W^\epsilon = \sigma^2 I, \sigma^2 = \tfrac12\tanh \epsilon$ and $\mu_{\v m}^{\epsilon}$ on the contracted square lattice~\eqref{Eq:contracted-square-lattice} and with $\mathcal{M}_k, s_k(\v m),c_{\v m}^\epsilon$ given in Eqs.~\eqref{Eq:app-MK-sets}-\eqref{Eq:app-sign-patterns}-\eqref{Eq:app-cm-weights}, respectively.

Let $M_o^{(b)}$ be the two-outcome POVM element corresponding to homodyne detection of the rotated quadrature $\hat{q}_{\theta_b}$ followed by periodic binning into the sets $R_0, R_1$~[Eq.~(\ref{Eq:binning})],
\begin{equation}
    M_o^{(b)} = \int_{R_o} \textrm{d}q_{\theta_b} \ketbra{q_{\theta_b}} \quad \text{with} \quad o \in \{0,1\}.
\end{equation}
Then, the single-mode overlap admits the explicit series representation
\begin{equation}
    t^{(b)}_{k,o} = (-1)^{\delta_{k,2}} \sum_{\v{m}\in\mathcal M_k} c_{\v{m}}^\epsilon\, s_k(\v{m}) \frac12 \sum_{s\in\mathbb Z} \left[ \mathrm{erf}\left(\frac{(2s+o)\sqrt{\pi}+\frac{\sqrt{\pi}}{2}-\mu_{\theta_b,\v{m}}}{\sqrt{2}\,\sigma}\right) - \mathrm{erf}\left(\frac{(2s+o)\sqrt{\pi}-\frac{\sqrt{\pi}}{2}-\mu_{\theta_b,\v{m}}}{\sqrt{2}\,\sigma}\right) \right], \end{equation}
where $\mu_{\theta_b,\v m}$ is the projection on the lattice mean onto the measured quadrature direction $\mu_{\theta_b,\v m} = (\cos \theta_b, \sin \theta_b).\mu_{\v m}^{\epsilon}$ and the series converges absolutely.
\end{thm}
\begin{proof}
    We start by expressing $t^{(\theta)}_{k,o}$ as a phase-space integral and use Corollary~\ref{Cor:app-rotated-wigner-marginal} with $A=\sigma_k^\epsilon$:
\begin{equation}\label{eq:t_as_phase_space_integral}
t^{(\theta)}_{k,o} = \int_{R_o} \textrm{d}q_\theta
\int_{-\infty}^{\infty} dp_\theta
W_{\sigma_k^\epsilon}\left(q_\theta\cos\theta - p_\theta\sin\theta, q_\theta\sin\theta + p_\theta\cos\theta\right).
\end{equation}
Substituting Eq.~\eqref{Eq:mix} into \eqref{eq:t_as_phase_space_integral} gives
\begin{align}\label{Eq:app-mix-integral-sum}
t^{(\theta)}_{k,o} =
(-1)^{\delta_{k,2}}
\sum_{\v{m}\in\mathcal M_k}
c_{\v{m}}^\epsilon\, s_k(\v{m})\,
\int_{R_o} \textrm{d}q_\theta
\int_{-\infty}^{\infty} dp_\theta
G_{\v{\mu}_{\v{m}}^\epsilon,\sigma^2\iden}
\left(q_\theta\cos\theta - p_\theta\sin\theta, q_\theta\sin\theta + p_\theta\cos\theta\right).
\end{align}
The exchange of the sum and integrals in Eq.~\eqref{Eq:app-mix-integral-sum} is justified by the absolute convergence of the series, guaranteed by the Gaussian envelope $c_{\v m}^\epsilon$. Since the covariance is isotropic, the Gaussian kernel satisfies the orthogonal transformation property: for any orthogonal matrix $O$, $G_{\v\mu,\sigma^2\iden}(O\v x) = G_{O^T\v{\mu},\sigma^2\iden}(\v x)$. Taking $O = S(\theta)^{-1}$ (which is orthogonal) and noting that $O^\top = S(\theta)$, we define the rotated mean $\v\mu_{\v m,\theta}^\epsilon := S(\theta)\v\mu_{\v m}^\epsilon$. Then, $G_{\v{\mu}_{\v{m}}^\epsilon,\sigma^2\iden}
\left(q_\theta\cos\theta - p_\theta\sin\theta, q_\theta\sin\theta + p_\theta\cos\theta\right)=
G_{\v{\mu}_{\v{m},\theta}^\epsilon,\sigma^2\iden}(q_\theta,p_\theta)$. Writing $\v{\mu}_{\v m,\theta}^\epsilon=(\mu_{\theta,\v m},\nu_{\theta,\v m})^\top$ and using the explicit Gaussian kernel, 
\begin{equation}
G_{\v{\mu}_{\v{m},\theta}^\epsilon,\sigma^2\iden}(q_\theta,p_\theta)=
\frac{1}{2\pi\sigma^2}
\exp\left[-\frac{(q_\theta-\mu_{\theta,\v{m}})^2+(p_\theta-\nu_{\theta,\v{m}})^2}{2\sigma^2}\right].
\end{equation}
Integrating over $p_\theta$ gives the one-dimensional Gaussian density in $q_\theta$:
\begin{equation}
    \int_{-\infty}^{\infty} dp_\theta G_{\v{\mu}_{\v{m},\theta}^\epsilon,\sigma^2\iden}(q_\theta,p_\theta) = \frac{1}{\sqrt{2\pi}\sigma} \exp\left[-\frac{(q_\theta-\mu_{\theta,\v{m}})^2}{2\sigma^2}\right] =:\varphi(q_\theta;\mu_{\theta,\v{m}},\sigma^2).
\end{equation}
Note that the scalar $\v\mu_{\theta,\v m}$ is the projection of $\mu_{\v m}^{\epsilon}$ onto the measured quadrature direction $\mu_{\theta,\v{m}} = (\cos\theta, \sin \theta) . \v{\mu}_{\v m}^\epsilon$. Using the contracted-lattice form of $\v{\mu}_{\v m}^\epsilon$, this becomes $\mu_{\theta,\v m}
=
\sech(\epsilon)\,\frac{\sqrt{\pi}}{2}\left(m_1\cos\theta+m_2\sin\theta\right)$. Hence, the single-mode overlap is reduced to a 1-D Gaussian probability over $R_o$, i.e., 
\begin{align}\label{Eq:app-t-integral}
   t^{(\theta)}_{k,o} &= (-1)^{\delta_{k,2}} \sum_{\v m\in\mathcal M_k} c_{\v m}^\epsilon s_k(\v m) \int_{R_o} \textrm{d}q_\theta\;\varphi(q_\theta;\mu_{\theta,\v m},\sigma^2) = (-1)^{\delta_{k,2}} \sum_{\v m\in\mathcal M_k} c_{\v m}^\epsilon s_k(\v m)\sum_{s\in\mathbbm{Z}} \int_{r^{(-)}_{s,o}}^{r^{(+)}_{s,o}}\textrm{d}q \,\varphi(q;\mu_{\theta,\v m},\sigma^2),
\end{align}
where we used the fact that the integral over a countable union of disjoint sets equals the sum of integrals over each set $\int_{R_o} \textrm{d}q\,f(q) = \sum_{s\in \mathbbm{Z}} \int_{r^{(-)}_{s,o}}^{r^{(+)}_{s,o}} \textrm{d}q \, f(q)$. The integration limits are $r^{(\pm)}_{s,o} = (2s+o)\sqrt{\pi}\pm\frac{\sqrt{\pi}}{2}$. Note that the integral in Eq.~\eqref{Eq:app-t-integral} can be easily computed. Let $u=(q-\mu)/\sigma$, it follows that
\begin{equation}
    \int_{r^{(-)}}^{r^{(+)}} dq\;\varphi(q;\mu,\sigma^2)=\frac{1}{\sqrt{2\pi}}\int_{(r^{(-)}-\mu)/\sigma}^{(r^{(+)}-\mu)/\sigma}\textrm{d}u \,e^{-u^2/2}
=\Phi\left(\frac{r^{(+)}-\mu}{\sigma}\right)-\Phi\left(\frac{r^{(-)}-\mu}{\sigma}\right),
\end{equation}
where $\Phi(x)=\frac{1}{\sqrt{2\pi}}\int_{-\infty}^{x}e^{-u^2/2}\,du$ is the standard normal cumulative distribution function. Finally, using $\Phi(x)=\tfrac12\left[1+\mathrm{erf}\left(\frac{x}{\sqrt2}\right)\right]$, the single-mode overlap can be expressed as an error-function series
\begin{equation}
t^{(b)}_{k,o} = (-1)^{\delta_{k,2}} \sum_{\v m\in\mathcal M_k} c_{\v m}^\epsilon\, s_k(\v m)\; \frac12 \sum_{s\in\mathbb Z} \left[ \mathrm{erf}\left(\frac{(2s+o)\sqrt{\pi}+\frac{\sqrt{\pi}}{2}-\mu_{\theta_b,\v m}}{\sqrt{2}\,\sigma}\right) - \mathrm{erf}\left(\frac{(2s+o)\sqrt{\pi}-\frac{\sqrt{\pi}}{2}-\mu_{\theta_b,\v m}}{\sqrt{2}\,\sigma}\right) \right].
\end{equation}
This completes the proof of the theorem.
\end{proof}
A remark is that the inner sum over $s$ is an exact infinite series representation of the integral over $R_o$. Its rapid convergence follows from the Gaussian decay of $\varphi(q;\mu_{\theta_b, \v m},\sigma^2)$, which suppresses contributions from intervals far from the mean.

\end{document}